\documentclass[conference]{IEEEtran}
\IEEEoverridecommandlockouts
\usepackage{amsfonts}
\usepackage{dsfont} 
\usepackage{setspace}
\usepackage{color}
\usepackage{amssymb}
\usepackage{cite}
\usepackage[cmex10]{amsmath}
\usepackage{amsthm}
\usepackage{algorithm}
\usepackage{algorithmic} % algorithm
\usepackage{array}
\usepackage{mathrsfs}
\usepackage{graphicx}
\usepackage{latexsym}
\usepackage{amscd}
\usepackage{amsfonts}
\usepackage{subfigure}
\usepackage{amsmath,amscd,amssymb,verbatim}
\usepackage{graphics}
\usepackage{amsthm}
\usepackage[T1]{fontenc}
\usepackage[utf8]{inputenc}
\usepackage{authblk}
\usepackage[nocomma]{optidef}
\usepackage{mathtools}
\usepackage{anyfontsize} 
\usepackage{soul}
\usepackage{textcomp}
\usepackage{xcolor}
\usepackage{blindtext}
\usepackage[left=0.673in,right=0.673in,top=0.71in,bottom=1in]{geometry}
\IEEEoverridecommandlockouts

\IEEEpubid{%
	\parbox[b]{\columnwidth}{%
		\footnotesize
		\textcopyright~2026 IEEE. Personal use of this material is permitted.
		Permission from IEEE must be obtained for all other uses, in any current
		or future media, including reprinting/republishing this material for
		advertising or promotional purposes, creating new collective works, for
		resale or redistribution to servers or lists, or reuse of any copyrighted
		component of this work in other works.%
	}%
	\hspace{\columnsep}%
	\makebox[\columnwidth]{}%
}

\newtheorem{lemma}{Lemma}

\newtheorem{corollary}{Corollary}
\usepackage{color}

\begin{document}

\title{AoI-Oriented Globally Optimal Joint Source and Update Scheduling in Fluid Antenna Systems\vspace{-.2cm}}
\author{
% %%\vspace{-.3cm}
 Xiaopeng Yuan$^{\dag}$,  Paul Zheng$^{\dag}$ and Anke Schmeink$^{\dag,\ddag}$ \vspace{0.1cm} \\ 
 $^\dag$Chair of Information Theory and Data Analytics, RWTH Aachen University, Germany\\ Email: $yuan|zheng|schmeink$@inda.rwth-aachen.de\\
    $^{\ddag}$Cluster of Excellence CARE, TU Dresden and RWTH Aachen, Germany\vspace{-0.3cm}
%\thanks{This work was supported by the BMFTR Germany projects 6GEM+ under Grant 16KIS2409K and GEM-X under Grant 16KISS004K.}
\thanks{\vspace{0.5cm}}
}
%\vspace{-0.4cm}}
\maketitle
%\vspace{-0.3cm}    

\begin{abstract}
As a promising technique, fluid antennas enable adaptive radio environment management and interference mitigation through reconfigurable fluid port selections. 
In this work, to explore the benefits of fluid antennas for data freshness enhancement, we consider a fluid-antenna assisted status update system supported by multiple source nodes monitoring the same environmental status.  
We assume a subset of the source nodes are activated to report status updates with different periods. Each user is assigned to one source node and equipped with an fluid antenna to adaptively enhance the channel gain to assigned source node while mitigating interference from other activated source nodes. %we formulate a maximum age of information (AoI) minimization problem, where 
With maximal signal-to-interference-noise ratio (SINR)-based fluid port selection at all users, we formulate an optimization problem to minimize the maximum average age of information (AoI), where source node activation and assignment, i.e., source scheduling, is jointly optimized with the update periods of all activated source nodes. To optimally solve the resulting mixed-integer nonlinear problem, we first consider given source scheduling decision and apply performance achievability analysis. Aided by fixpoint theory, we propose an efficient bisection algorithm for optimal update scheduling. Based on these characterizations, we further propose a filtering algorithm which efficiently eliminates all non-optimal source scheduling decisions. The globally optimal joint solution is then obtained by combining the resulting optimal source scheduling  with its corresponding optimal update scheduling. Numerical results validate the optimality of proposed solution and demonstrate the effectiveness of fluid antennas in enhancing data freshness.
\end{abstract}


%\vspace{-0.1cm}    
\section{Introduction}
%\vspace{-0.03cm}

Towards advanced mission-critical networks, data freshness plays a critical role in enabling real-time perception and timely environment awareness~\cite{aoi_1,aoi_2}. 
Especially for mission-critical real-time applications, outdated information may cause erroneous operations and misleading decisions~\cite{aoi_3}, resulting in unpredictable risk and increased costs. 
To deal with this issue, age of information (AoI) has been proposed and widely adopted as a key metric for quantifying data freshness. 
Abundant research has been conducted on minimizing AoI and enhancing data freshness through update scheduling~\cite{aoi_5} and status packet management~\cite{aoi_6}. %, and resource allocation~\cite{aoi_7}. 
In parallel, extensive AoI investigations have been carried out across various scenarios, including mobile edge computing (MEC) networks~\cite{aoi_8}, relaying systems~\cite{aoi_9} and secure communication networks~\cite{aoi_10}. 

Moreover, the high dynamics of wireless networks, together with the coexistence of a wide range of wireless services, induce a complex radio environment with severe interference~\cite{interference}, rendering reliable data refreshing more challenging in wireless status update systems. To tackle this challenge,
fluid antenna systems (FAS) %, as one of the promising flexible antenna technologies, 
have recently been  discovered to exhibit strong capabilities in adaptively reconfiguring the radio environment and thereby mitigating the interference~\cite{fluid_1,fluid_2}. Particularly, in FAS, the position of fluidic element can be flexibly adjusted across multiple preset fluid ports, which enables rapid selection of the most favorable fluid port against deep fading and interference. Leveraging the potential of fluid antennas, a new fluid antenna multiple access (FAMA) scheme has been developed in \cite{fluid_3,fluid_4} towards next-generation multi-access networks, while fluid antenna assisted multi-input multi-output (MIMO) systems have been investigated in \cite{fluid_5}.

Nevertheless, despite the extensive research efforts, the integration of fluid antennas into status update system for AoI minimization has not yet been fully explored. 
A recent study~\cite{fluid_6} investigated fluid-antenna aided AoI minimization in an unmanned aerial vehicle (UAV)-assisted wireless network where the UAV acted as the single status source, while  the interference mitigation capability of fluid antennas was overlooked.
For a practical resilient status update system, more than one status source node can be accessible and activated simultaneously. In such a case, the coexisting source nodes lead to severe mutual interference, which can potentially be largely mitigated through deployed fluid antennas. Meanwhile, when multiple users request status updates simultaneously, source scheduling for each receiving user becomes more challenging due to the integration of fluid antennas. So far, to the best of our knowledge, investigations of fluid antenna for multi-source interfered status update systems, as well as the optimal joint source and update scheduling for AoI minimization, are still missing in the existing literature. 

To fill these gaps, we study a fluid-antenna assisted status update system, where multiple fluid-antenna users require fresh status updates from multiple source nodes monitoring the same status process. Taking into account the interference mitigation capabilities of fluid antennas, we aim to minimize the maximum average AoI among all users via jointly optimizing the source node activation, assignment and update periods for all activated source nodes. To optimally tackle the formulated mixed-integer nonlinear problem, we propose a fixpoint-theorem based algorithm for the optimal update period design, which then enables us to efficiently filter out the optimal decision of source node activation  and assignment. The main contributions of this work are summarized as follows:
\begin{itemize}
    \item In this work, we investigate a multi-source fluid-antenna-assisted status update system. The potential of fluid antennas in mitigating interference and simultaneously minimizing AoI is demonstrated.
    \item To tackle the analytical challenges introduced by fluid antennas and finite blocklength codes, we propose a novel optimization approach, where performance achievability is characterized via fixpoint theory, thereby enabling optimal update period design. 
    \item Based on the %performance achievability 
    characterizations, we propose an efficient algorithm to determine the optimal source node activation and assignment, eventually achieving the globally optimal joint source and update scheduling. The entire design framework %methodology 
    can be extended to a wide range of %numerous 
    problems with high analytical complexity for the optimal solutions. 
\end{itemize}

The remaining sections are organized as follows: In Section II, we state the optimization problem. Section III describes the optimal update period design, based on which the optimal design for source node activation and assignment is proposed in Section IV. Finally, the proposed optimal design is evaluated in Section V and the entire work is concluded in Section VI.

\vspace{-0.1cm}
\section{Problem Statement}
%\vspace{-0.1cm}

\begin{figure}[t]
    \centering
\includegraphics[width=.75\linewidth,trim=0 20 0 10]{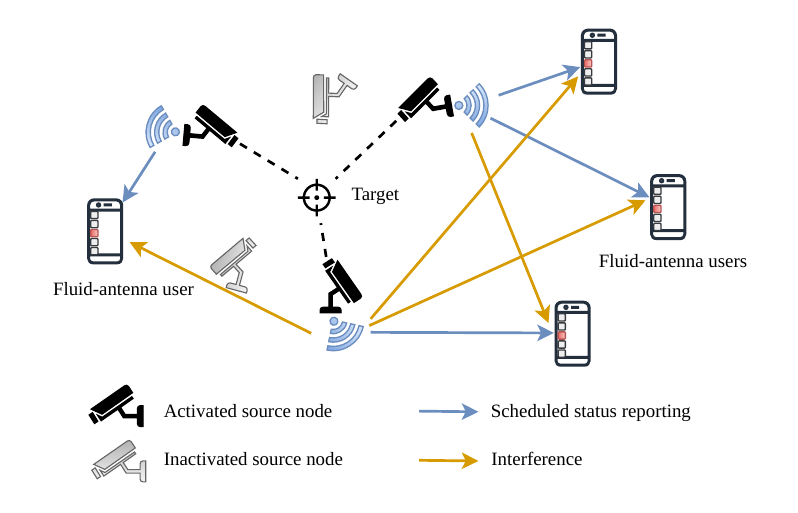}
    \caption{Multi-source status update in fluid antenna systems.}
    \label{fig:scenario}
    \vspace{-0.5cm}
\end{figure}

\subsection{System Model}
In this work, we study a wireless status update system, where multiple users equipped with fluid antennas require fresh information from multiple available source nodes monitoring the same status. We denote by $K$ the total user number and by $S$ the number of source nodes. In practical applications, the multiple source nodes may correspond to multiple monitors observing the same target as illustrated in Fig.~\ref{fig:scenario}, or sensors measuring the same environmental parameter, such as humidity. We assume $U$ out of the $S$ source nodes are activated for status reporting. Let $\mathcal A_S\triangleq\{1,...,S\}$ denote the set of all available source nodes and let $\mathcal A_U \subseteq \mathcal A_S$ denote the set of activated source nodes. Clearly, we have $U\leq S$.

The activated source nodes can explore their spatial diversity to enable more efficient status reporting to multiple users. For update reception, each user is assigned to one activated source node and leverages the flexible reconfiguration capability of  fluid antenna to enhance reception reliability. Specifically, for each user $k\in\mathcal K\triangleq \{1,...,K\}$, we denote by $\alpha(k)\in\mathcal A_U$ the assigned active source node. Note that multiple users may be assigned to the same source node, exploiting the broadcast nature of wireless status updating as implied in Fig.~\ref{fig:scenario}.

Moreover, we assume the fluid antenna equipped at receiving users has $N$ preset fluid ports evenly distributed along a line of length $\lambda_c W$, where $\lambda_c$ is the carrier wavelength. Let $h_{s,k,n}$ denote the channel gain from source node $s$ to $n$-th fluid port at user $k$, where $s\!\in\!\mathcal A_S$ and $n\in\mathcal N\triangleq\{1,...,N\}$. Since the closely placed fluid ports experience correlated channels~\cite{fluid_3}, we consider a rich scattering environment and model the cross-correlation between channel gains  $h_{s,k,n_1}$ and $h_{s,k,n_2}$ as\vspace{-0.2cm}
\begin{equation}\vspace{-0.1cm}
    \phi_{s,k}(n_1-n_2)=\frac{\delta_{s,k}^2}{2}J_0(\frac{2\pi(n_1-n_2)}{N-1}W),~n_1,n_2\in\mathcal N,
\end{equation}
where $\delta_{s,k}^2$ represents the variance of channel from source node $s$ to arbitrary fluid port of user $k$, and $J_0(\cdot)$ is the zero-order Bessel function of the first kind. 
%In this work, we assume different activated source nodes perform status reporting independently, as tight synchronization among multiple distributed source nodes is difficult and costly to achieve in practice. Therefore, for each user $k$, the activated but non-selected source nodes act as interference during the status reception. 
Denoting by $P$ the transmit power of source nodes, the signal-to-interference-plus-noise ratio (SINR) at the $n$-th port of user $k$ for the signal transmitted from its assigned source node $\alpha(k)$ can be expressed as\vspace{-0.1cm}
\begin{equation}\vspace{-0.1cm}
    \gamma_{\alpha(k),k,n}(\mathcal A_U,\mathbf h)=\frac{Ph_{\alpha(k),k,n}}{\sum\nolimits_{j\in\mathcal A_U,j\neq \alpha(k)}Ph_{j,k,n}+\sigma^2},%\label{eq:sinr}
\end{equation}
where $\mathbf h$ denotes the vector containing all channel gains $h_{s,k,n}$ and $\sigma^2$ is the noise power. %additive white Gaussian noise (AWGN). 
By further exploiting the reconfiguration capability of fluid antennas at users, the received SINR at each user $k$ can be enhanced by selecting the highest SINR among all fluid ports~\cite{fluid_3}, i.e.,\vspace{-0.1cm}
\begin{equation}\vspace{-0.1cm}
    \hat\gamma_{\alpha(k),k}(\mathcal A_U,\mathbf h)=\max_{n\in\mathcal N}\Big\{\frac{Ph_{\alpha(k),k,n}}{\sum\limits_{j\in\mathcal A_U,j\neq \alpha(k)}Ph_{j,k,n}+\sigma^2}\Big\}.\label{eq:sinr}
\end{equation}

Note that in this work, we assume different activated source nodes perform status reporting independently, as tight synchronization among multiple distributed source nodes is difficult to achieve in practice. Therefore, for each user $k$, the activated but non-selected source nodes act as interference during the status reception, as indicated in \eqref{eq:sinr} and also in Fig.~\ref{fig:scenario}. 
To achieve fresher and more reliable status monitoring, careful design of source node activation $\mathcal A_U$ and source assignment $\alpha(\cdot)$ is thus essential. 
In addition, the adoption of fluid antennas enables adaptive SINR enhancement through flexible fluid port selection as shown in \eqref{eq:sinr}, but also introduces additional analytical complexity in the system design. %corresponding system analysis and optimization. 

\subsection{AoI Characterization}
\vspace{-0.1cm}
Next, to well represent the data freshness, we adopt AoI %characterization. %Since maintaining data freshness is critical in status update systems, the concept of AoI has been widely adopted 
as the key performance metric. 
At any time point $t$, the instantaneous AoI is defined as \vspace{-0.1cm}
\begin{equation}\vspace{-0.1cm}
    \Delta (t)=t-u(t),
\end{equation}
where $u(t)$ represents the time for the most recently received status update being generated. The time-average AoI can be subsequently derived as \vspace{-0.15cm}
\begin{equation}\vspace{-0.1cm}
    \mathbb E[\Delta]=\limsup_{\tau\to\infty}\frac{1}{\tau}\int_0^\tau \Delta(t)dt.
\end{equation}

In our considered status update system, we assume each activated source node $\alpha(k)$ periodically reports the status update with a period of $T_{\alpha(k)}$, resulting in $m_{\alpha(k)}=\frac{T_{\alpha(k)}}{T_s}$ available symbols for wireless transmissions, where $T_s$ denotes the symbol duration. 
According to \cite{fbl_0}, over a finite blocklength of $m$, i.e., $m$ available symbols, the achievable decoding error probability of a wireless transmission is determined by the SINR $\gamma$ and data packet size $D$ as follows\vspace{-0.1cm}
\begin{equation}\vspace{-0.1cm}
    \varepsilon(m,\gamma)=Q\Big(\sqrt{\frac{m}{V(\gamma)}}(\log_2(1+\gamma)-\frac{D}{m})\ln2\Big),\label{eq:fbl}
\end{equation}
where $V(\gamma)=1-\frac{1}{(1+\gamma)^2}$ represents the channel dispersion and $Q(\cdot)$ denotes the Gaussian $Q$-function. Accordingly, with packet size $D$ for each status update, the transmission error probability for the status update 
from activated source node $\alpha(k)$ to user $k$ with channel gain $\mathbf h$ can be expressed as $\varepsilon(m_{\alpha(k)},\hat\gamma_{\alpha(k),k}(\mathcal A_U,\mathbf h))$.
%\begin{equation}
 %   \varepsilon_k=\varepsilon(m_{\alpha(k)},\hat\gamma_{\alpha(k),k}(\mathcal A_U,\mathbf h)).
%\end{equation}
With time-varying channel conditions, the average transmission error probability for each status update to user $k$ can be derived as\vspace{-0.1cm}
\begin{equation}\vspace{-0.1cm}
    \mu_k(\mathcal A_U,\alpha(\cdot),m_{\alpha(k)})=\int_{\mathbf h}\varepsilon(m_{\alpha(k)},\hat\gamma_{\alpha(k),k}(\mathcal A_U,\mathbf h))\psi(\mathbf h)\text{d}\mathbf h,\label{eq:error}
\end{equation}
where $\psi(\mathbf h)$ is the probability density function (PDF) of channel gain vector $\mathbf h$. 

Then, %Next, we characterize the average AoI based on the derived average transmission error probability. L
let $Y_k$ denote the time difference between two consecutive successful status reception at user $k$, which follows a random distribution. Specifically, the probability for event $Y_k=zT_{\alpha(k)}=zm_{\alpha(k)}T_s$ with $z\in\mathbb N_+$ can be obtained as %\vspace{-0.1cm}
%\begin{equation}\vspace{-0.1cm}
    $\text{Pr}(Y_k=zT_{\alpha(k)})=\mu_k^{z-1}(1-\mu_k)$.
%\end{equation}
According to \cite{aoi_formula}, the average AoI for user $k$ can be characterized as\vspace{-0.1cm}
\begin{align}
    \mathbb E[\Delta_k]&=\frac{\mathbb E[Y_k^2]}{2\mathbb E[Y_k]}+T_{\alpha(k)}\nonumber \\
    &=\frac{1}{2}m_{\alpha(k)}T_s+\frac{m_{\alpha(k)}T_s}{1-\mu_k(\mathcal A_U,\alpha(\cdot),m_{\alpha(k)})}.
\end{align}
%Clearly, the average AoI for each user is jointly determined by the source node activation $\mathcal A_U$, source selection $\alpha(\cdot)$ and update period $m_{\alpha(k)}T_s$ of activated source nodes. 

\subsection{Problem Formulation}

In this work, taking user fairness into account, we aim to enhance data freshness for all users through the joint optimization of source node activation $\mathcal A_U$, source assignment $\alpha(\cdot)$ and the update blocklength $m_{\alpha(k)}$ of activated source nodes, where the maximal average AoI among all users is supposed to be minimized. The optimization problem can be formulated as\vspace{-.1cm}
\begin{align}
    \text{(P1)}\!\!\!\min_{\mathcal A_U,\alpha(\cdot),m_{\alpha(k)}} \!\!\!\!&~~\max_{k\in\mathcal K}\big\{\frac{1}{2}m_{\alpha(k)}T_s\!+\!\frac{m_{\alpha(k)}T_s}{\!1\!-\!\mu_k(\mathcal A_U,\alpha(\cdot),m_{\alpha(k)})}\big\}\\
    \text{s.t.}~~~~&~~ \mathcal A_U\subseteq \mathcal A_S, \\
    &~~\alpha(k)\in\mathcal A_U,\\
    &~~m_{\alpha(k)}>0,~\forall k\in\mathcal K.
\end{align}
Due to the structural impact of source node activation $\mathcal A_U$ and source assignment $\alpha(\cdot)$, as well as the complicated expression of average error probability $\mu_k(\mathcal A_U,\alpha(\cdot),m_{\alpha(k)})$,  the resulting joint optimization %for enhancing data freshness 
is highly challenging to solve optimally. In particular, the SINR expression \eqref{eq:sinr} which depends highly on optimized fluid port selection, and the integral involving $Q$-function in the average error probability \eqref{eq:error}, render the problem nonconvex and analytically intractable. Therefore, more adaptive optimization approach is highly demanded for optimally and effectively addressing problem~(P1).

\section{Fixpoint-Based Optimal Blocklength Design}

\subsection{Subproblem Formulation}
To optimally address the joint optimization problem (P1), we first characterize the optimal blocklength design with given source node activation and assignment. For any $\mathcal A_U$ and $\alpha(\cdot)$, the subproblem of blocklength design can be formulated as\vspace{-0.1cm}
\begin{align}
    \!\text{(P2)}\min_{m_{\alpha(k)}} \!\!\!\!&~~\max_{k\in\mathcal K}\big\{\frac{1}{2}m_{\alpha(k)}T_s\!+\!\frac{m_{\alpha(k)}T_s}{\!1\!-\!\mu_k(\mathcal A_U,\alpha(\cdot),m_{\alpha(k)})}\big\}\\
    \text{s.t.}&~~m_{\alpha(k)}>0,~\forall k\in\mathcal K,
\end{align}

\vspace{-0.2cm}
\noindent{}which remains nonconvex due to the complex objective function and cannot be directly and optimally solved using conventional optimization techniques.

\subsection{Fixpoint-Based Inference of Performance Achievability}

Instead of directly optimizing the blocklength, we focus on performance achievability with respect to an objective threshold, leading to the following Lemma~\ref{le:threshold}.

\begin{lemma}\label{le:threshold}
    Given an objective threshold $\lambda>0$, there exists $m_{\alpha(k)}>0$ such that \vspace{-0.1cm}
    \begin{equation}\vspace{-0.1cm}
        \frac{1}{2}m_{\alpha(k)}T_s\!+\!\frac{m_{\alpha(k)}T_s}{\!1\!-\!\mu_k(\mathcal A_U,\alpha(\cdot),m_{\alpha(k)})}=\lambda,~\forall k\in\mathcal K,\label{eq:le_1}
    \end{equation}
    if and only if the threshold $\lambda$ is achievable for problem (P2), i.e., 
    the optimal objective $\eta^*$ for problem (P2) is below or equal to the threshold $\lambda$, resulting in  $\eta^*\leq \lambda$. 
%    Otherwise, the threshold $\lambda$ is unachievable for problem (P2), such that $\eta^*>\lambda$. 
\end{lemma}
\vspace{-0.3cm}
\begin{proof}
    For given $\lambda>0$, if there exists $m_{\alpha(k)}>0$ fulfilling~\eqref{eq:le_1}, we actually find a feasible solution $\{m_{\alpha(k)}\}$ for (P2) with objective value of $\lambda$, such that $\eta^*\leq \lambda$.
    
    Now, we assume $\eta^*\leq \lambda$ and construct positive $m_{\alpha(k)}$ fulfilling~\eqref{eq:le_1}. According to \eqref{eq:fbl} and \eqref{eq:error}, with a relatively small blocklength $m_{\alpha(k)}$, the average error probability $\mu_k(\mathcal A_U,\alpha(\cdot),m_{\alpha(k)})$ for user $k$ can easily reach~$1$. As a result, the average AoI for user $k$, i.e., the left-hand-side term of \eqref{eq:le_1}, approaches infinity as blocklength $m_{\alpha(k)}$ approaches $0$. Since the average AoI for user $k$ is continuous in $m_{\alpha(k)}>0$ and can reach the optimal objective level $\eta^*$, for $\eta^*\leq \lambda$, there must exist a positive blocklength $m_{\alpha(k)}$ such that the average AoI is equal to $\lambda$, i.e., \eqref{eq:le_1} holds, which proves the statement.
\end{proof}
\vspace{-0.3cm}
The above Lemma~\ref{le:threshold} provides a criterion for evaluating the performance achievability of problem (P2) with respect to $\lambda$. Next, we propose a fixpoint based inference strategy to determine the solution existence of \eqref{eq:le_1} for any given $\lambda>0$. 

We define an $|\mathcal A_U|$-dimension space $\Omega\triangleq [0,\infty]^{|\mathcal A_U|}$. Any blocklength vector $\mathbf m$ including all blocklength variables $m_{\alpha(k)}$ belongs to the space $\Omega$. We further introduce a partial order ``$\preceq$'' over space $\Omega$ as $\mathbf m\preceq \hat{\mathbf m}$ for $\mathbf m,\hat{\mathbf m}\in \Omega$ if and only if $m_{\alpha(k)}\leq \hat m_{\alpha(k)}$, $\forall \alpha(k)\in\mathcal A_U$. According to~\cite{fluid_4,routing_1}, $(\Omega,\preceq)$ can be easily validated as a complete lattice. Subsequently, based on \eqref{eq:le_1}, we define a mapping as follows\vspace{-0.1cm}
\begin{equation}\vspace{-0.1cm}
    \mathbf f_\lambda:\Omega\to\Omega,~m_{\alpha(k)}\mapsto \frac{\lambda}{T_s}\frac{1}{\frac{1}{2}+\frac{1}{1-\mu_k(\mathcal A_U,\alpha(\cdot),m_{\alpha(k)})}},
    \label{eq: fixpoint-mapping}
\end{equation}
which can be easily shown to be monotonic, i.e., $f_\lambda(\mathbf m)\preceq f_\lambda(\hat{\mathbf m})$ for any $\mathbf m\preceq\hat{\mathbf m}$, due to the monotonic decreasing property of $\mu_k(\mathcal A_U,\alpha(\cdot),m_{\alpha(k)})$ with respect to $m_{\alpha(k)}$. Clearly, the blocklength vector $\mathbf m$ satisfying \eqref{eq:le_1} is a fixpoint of mapping $\mathbf f_\lambda$, i.e., $\mathbf f_\lambda(\mathbf m)=\mathbf m$.

According to the Knaster-Tarski Fixpoint Theorem and top iteration in \cite{fixpoint}, we have the following Lemma~\ref{le:fixpoint}.
\vspace{-0.1cm}
\begin{lemma}\label{le:fixpoint}
    For complete lattice $(\Omega,\preceq)$ and monotonic mapping $\mathbf f_\lambda$, the set of fixpoints with respect to $\mathbf f_\lambda$ is a complete lattice under partial order ``$\preceq$'', and there exists a unique greatest fixpoint of $\mathbf f_\lambda$. 

    By starting from the maximal element $\top\triangleq[\infty,...,\infty]^T$ in $\Omega$ and applying mapping $\mathbf f_\lambda$ iteratively, the mapped element converges to the greatest fixpoint of $\mathbf f_\lambda$. 
\end{lemma}
\vspace{-0.1cm}
Based on Lemma~\ref{le:fixpoint}, we can obtain the greatest fixpoint of mapping $\mathbf f_\lambda$ for any $\lambda>0$ and discuss on two cases:
\begin{itemize}
    \item If the greatest fixpoint has any $m_{\alpha(k)}=0$, we can conclude that positive $\mathbf m$ satisfying \eqref{eq:le_1} does not exist, such that $\eta^*>\lambda$ according to Lemma~\ref{le:threshold}. Otherwise, we will construct a fixpoint larger than the greatest fixpoint.
    \item If the greatest fixpoint has all elements being positive, this fixpoint will be a feasible solution for \eqref{eq:le_1}, such that $\eta^*\leq \lambda$ according to Lemma~\ref{le:threshold}.
\end{itemize}

The above discussions enable a fixpoint-based inference on the performance achievability of problem (P2) with respect to any threshold $\lambda>0$, which motivates us to develop an efficient algorithm %in the following 
to efficiently determine the optimal objective $\eta^*$ for (P2), as well as the optimal blocklength. %to optimally address problem (P2).

\subsection{Optimal Blocklength Design}

At first, we initialize a performance interval $[\lambda_\text{min},\lambda_\text{max}]$ and define $\lambda=\frac{\lambda_\text{min}+\lambda_\text{max}}{2}$. Based on the threshold $\lambda$, we construct the mapping $\mathbf f_\lambda$ and obtain the greatest fixpoint based on Lemma~\ref{le:fixpoint}. If the greatest fixpoint contains positive elements only, the threshold $\lambda$ is then achievable for (P2) and acts as the upper bound for the optimal objective value $\eta^*$. We shorten the performance interval via updating $\lambda_\text{max}=\lambda$ and redefine $\lambda=\frac{\lambda_\text{min}+\lambda_\text{max}}{2}$ for further achievability inference. If the greatest fixpoint contains zero-value element, we then update $\lambda_\text{min}=\lambda$ and $\lambda=\frac{\lambda_\text{min}+\lambda_\text{max}}{2}$ for the next iteration. Via repeatedly shortening the performance interval, we eventually obtain the optimal objective value $\eta^*$ for problem (P2), while the correspondingly obtained greatest fixpoint with $\lambda=\eta^*$ is thus the optimal blocklength design. 

\section{Optimal Source Node Activation and Assignment}

So far, we have achieved optimal blocklength design for any given source node activation $\mathcal A_U$  and assignment $\alpha(\cdot)$. Although the optimal source node activation $\mathcal A_U$  and assignment  $\alpha(\cdot)$ can be determined by exhaustively  evaluating the optimal blocklength design over all possible decisions, this approach introduces significant computational complexity. Towards globally optimal and efficient source node activation and assignment, we propose an efficient algorithm fully leveraging the above conducted characterizations. % in the previous section. 

\subsection{Achievability Inference for Multiple Decisions}

For any threshold $\lambda>0$, the greatest fixpoint of $\mathbf f_\lambda$ can be readily computed, from which the performance achievability can be inferred. We denote by $\Pi$ the set of decision candidates for source node activation and assignment. Given the large number of candidate decisions, %for source node activation and selection, 
we can perform one-shot fixpoint inspections on all decision candidates under the same threshold $\lambda$, thereby obtaining explicit performance achievability for the original problem (P1). We have the following Corollary~\ref{le:filter} stating the criterion for the achievability inference.

\begin{corollary}\label{le:filter}
    Given a set of decision candidates $\Pi$, if a threshold $\lambda>0$ is achievable for  any candidate in $\Pi$, the threshold $\lambda$ is then no smaller than the optimal objective of problem (P1), and all candidates for which $\lambda$ is unachievable are non-optimal.

    Conversely, if a threshold $\lambda>0$ is unachievable for all candidates in $\Pi$, the threshold $\lambda$ is then strictly below the optimal objective of (P1).
\end{corollary}

\subsection{Optimal Decision Filtering}

Based on Corollary~\ref{le:filter}, we develop an efficient filtering strategy to obtain the optimal decision of source node activation and assignment. 
We first initialize the decision candidate set $\Pi$, a similar performance interval $[\lambda_\text{min},\lambda_\text{max}]$, and define $\lambda=\frac{\lambda_\text{min}+\lambda_\text{max}}{2}$. In each filtering round, we inspect the greatest fixpoints for all candidates in $\Pi$ based on threshold $\lambda$. If $\lambda$ is achievable for any candidate, we delete from $\Pi$ all non-optimal candidates that lead to zero-value element in their greatest fixpoints, and update $\lambda_\text{max}=\lambda$. In this case, the size of the candidate set $\Pi$ is reduced, thereby significantly lowering the computational complexity. 
If $\lambda$ is not achievable for all candidates in $\Pi$, we keep the set $\Pi$ and update the performance lower bound as $\lambda_\text{min}=\lambda$.

Via repeating the filtering rounds, the size of set $\Pi$ is continually reduced and eventually only one decision candidate is kept in the set, which is then the globally optimal decision for source node activation and assignment. Note that as non-optimal candidates are continually removed from set $\Pi$, fewer and fewer candidates need to be inspected in subsequent filtering rounds, leading to a significant reduction in computational complexity. 
Finally, based on the optimal decision candidate and updated performance interval, 
we can obtain the corresponding optimal blocklength design via applying the bisection algorithm in the previous section. The globally optimal joint solution for original problem (P1) is eventually obtained. The algorithm flow is summarized in Algorithm~\ref{al:optimal}. %, minimizing the maximal average AoI among all users. 

\begin{algorithm}[!t]%\small
	{\small
	\algsetup{linenosize=\large}
	%\vspace{.05in}
	\caption{\bf{--- Efficient Optimal Decision Filtering}}
	\begin{algorithmic}
		%\STATE \noindent{\bf{$\!\!\!\!\!\!$Initialization}} \\
		\STATE \!\!\!\!\!\noindent{\bf{a)}}~ Initialize decision candidate set $\Pi$ and interval $[\lambda_\text{min},\lambda_\text{max}]$.
		%\STATE \noindent{ \bf{$\!\!\!\!\!\!\!\!$Loop}} \\
		\STATE \!\!\!\!\!\noindent{\bf{b)}}~ Define $\lambda=\frac{\lambda_\text{min}+\lambda_\text{max}}{2}$.
		\STATE \!\!\!\!\!\noindent{\bf{c)}}~ Inspect greatest fixpoints for all candidates in $\Pi$ based on $\lambda$.
		\STATE \!\!\!\!\!\noindent{\bf{d)}}~ {\bf If} there exists all-positive greatest fixpoint,	
		\STATE \!\!\!\!\!~~~~~~~ Delete from $\Pi$ candidates with zero values in their fixpoints. 
		\STATE \!\!\!\!\!~~~~~~~ $\lambda_\text{max}=\lambda$. %, go back to {\bf a)}.
		\STATE \!\!\!\!\!~~~ \noindent{\bf{Else}}
		\STATE \!\!\!\!\!~~~~~~~ $\lambda_\text{min}=\lambda$.
		\STATE \!\!\!\!\!~~~ \noindent{\bf{End}}
        \STATE \!\!\!\!\!\noindent{\bf{e)}}~ {\bf If} set $\Pi$ contains more than one candidates,  
		\STATE \!\!\!\!\!~~~~~~~ Back to {\bf b)}.
		\STATE \!\!\!\!\!~~~ \noindent{\bf{Else}} 
		\STATE \!\!\!\!\!~~~~~~~ Determine optimal source scheduling as the only one in $\Pi$.
        \STATE \!\!\!\!\!~~~~~~~ Determine  optimal blocklengths based on Section III-C.
		\STATE \!\!\!\!\!~~~ \noindent{\bf{End}}
	\end{algorithmic}
	\label{al:optimal}
	%\vspace{-.1cm}
	}
\end{algorithm}

\section{Numerical Results}

In this section, we validate the proposed joint optimization framework %for joint source node activation, assignment, and update blocklength design, 
and evaluate the obtained optimal solution. %validate the analytical results developed in the previous sections.
%\subsection{Simulation Setup}
We consider a square deployment area of side length~$100$~m, in which~$S\!=\!7$ source nodes and~$K\!=\!4$ receivers are placed randomly. %uniformly at random.  
The carrier frequency~$f_c$ is set to~$2.4$~GHz and the path loss exponent is $2.3$ for large-scale channel gain modeling. 
Furthermore, we set by default  $P\!=\!10$\,dBm, $\sigma^2\!=\!-60$\,dBm, $N\!=\!10$, $D\!=\!256$ bits, $T_s\!=\!4\, \mu$s, and normalized fluid antenna length $W=1$. 
The average error probability $\mu_k(\mathcal A_U,\alpha(\cdot),m_{\alpha(k)})$ in \eqref{eq:error} is calculated through $1000$ independent channel realizations.

%\subsection{Validation of the Fixpoint-Based Achievability Inference}
\begin{figure}[t]
    \centering
\includegraphics[width=0.74\linewidth]{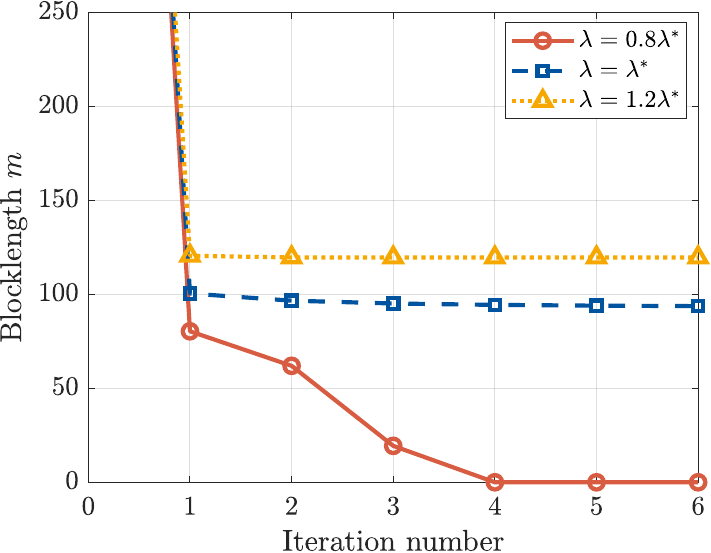}
    \vspace{-.2cm}
    \caption{Fixpoint convergence with different thesholds $\lambda$.}
    \label{fig: fixpoint}
    \vspace{-0.5cm}
\end{figure}

We first validate the fixpoint-based performance achievability inference for optimal blocklength designs presented in Section~III. We study one representative source node activation and assignment pair $(\mathcal A_U,\alpha(\cdot))$, % (for example the optimal one), 
and denote its optimal by $\lambda^*$. Then, we apply the top iteration of mapping $\mathbf f_\lambda$ in~\eqref{eq: fixpoint-mapping} for three test thresholds $\lambda\in\{0.8\lambda^*,\lambda^*,1.2\lambda^*\}$. 
Fig.~\ref{fig: fixpoint} presents the convergence behaviour of mapped elements. As indicated by Lemma~\ref{le:fixpoint}, the updated blocklengths are monotonically decreasing and converge to the greatest fixpoint of $\mathbf f_\lambda$ within very few iterations. 
For $\lambda=0.8\lambda^*<\lambda^*$, the blocklength collapse towards zero in four iterations, which by Corollary~\ref{le:filter} confirms the threshold is unachievable. %The case $\lambda=\lambda^*$ converges precisely to the optimal blocklength $\mathbf m^*$. 
For both achievable thresholds $\lambda=\lambda^*$ and $\lambda=1.2\lambda^*>\lambda^*$, the blocklengths converge to strictly positive values, which aligns with our characterizations on the performance achievability. %blocklength is achieved, since~$\lambda$ is achievable. 

%\subsection{Validation of the Filtering Algorithm}

\begin{figure}[t]
    \centering
    \includegraphics[width=0.74\linewidth]{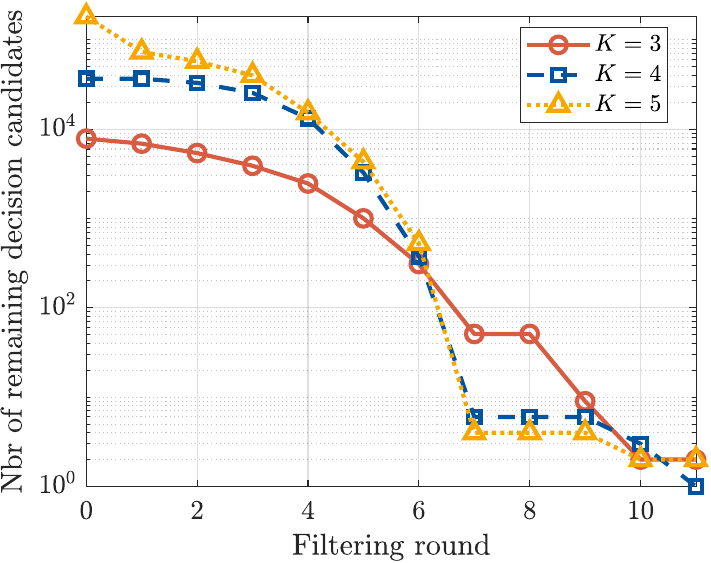}
    \vspace{-.3cm}
    \caption{Number of remaining decision candidates over filtering rounds.}
    \label{fig: filtering}\vspace{-0.1cm}
\end{figure}

Next, we validate the filtering strategy proposed in Section~IV by tracking the number of remaining decision candidates in $\Pi$ across filtering rounds. 
The number of remaining decision candidates %as a function of the filtering round 
is shown in Fig.~\ref{fig: filtering} with different user number $K$.
As indicated, more users~$K$ produce a larger initial candidate set. For all evaluated cases, we observe that the number of remaining decision pairs drops under 100 within seven iterations, which validates the filtering efficiency. %The preceding reduction is also clearly visible on the logarithmic scale. 
The algorithm converges when only the optimal decision is left in the set.
Note that for $K=3,5$, two decision candidates are kept finally, as their optimal objective values are sufficiently close such that both decisions can be considered as the optimal under our preset relative tolerance, i.e., $0.1\%$. % Specifically, the difference in the achievable objective is below $0.1\%$, therefore numerically indistinguishable.

%\subsection{System Evaluation}

\begin{figure}[t]
    \centering
    \includegraphics[width=0.74\linewidth]{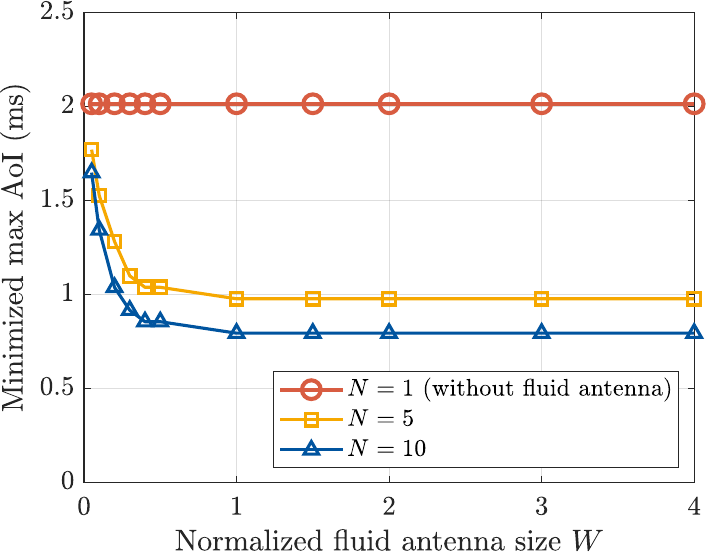}
    \vspace{-.2cm}
    \caption{Optimized AoI performance with different antenna sizes $W$ and fluid port numbers~$N$.}
    \label{fig: fig_wrt_W_N}
    \vspace{-0.5cm}
\end{figure}

Afterwards, we study the impact of the fluid-antenna parameters on the optimized average AoI. Fig.~\ref{fig: fig_wrt_W_N} shows the impact of the fluid antenna size $W$ on the minimized average AoI for different fluid port number~$N$, where $N=1$ corresponds to the user configuration without fluid antennas. %receiver and serves as a reference.
For case~$N=1$, the curve is flat in~$W$, since a single port does not benefit from spatial reconfiguration. As observed, the introduction of fluid antennas (i.e., $N>1$) has substantially improved the AoI performance, which validates the research necessity of this work. 
For cases with fluid antennas, the AoI decreases sharply as $W$ grows, reflecting the fact that the fluid ports are more decorrelated as the antenna aperture widens, so that the port-selection diversity in~\eqref{eq:sinr} provides substantial SINR gains. When $W$ is sufficiently large, % Beyond $W\!\approx\!1$~$\lambda$, 
the curves flatten, indicating that the effect of wider fluid port is essentially saturated.

\begin{figure}[t]
    \centering
\includegraphics[width=0.74\linewidth]{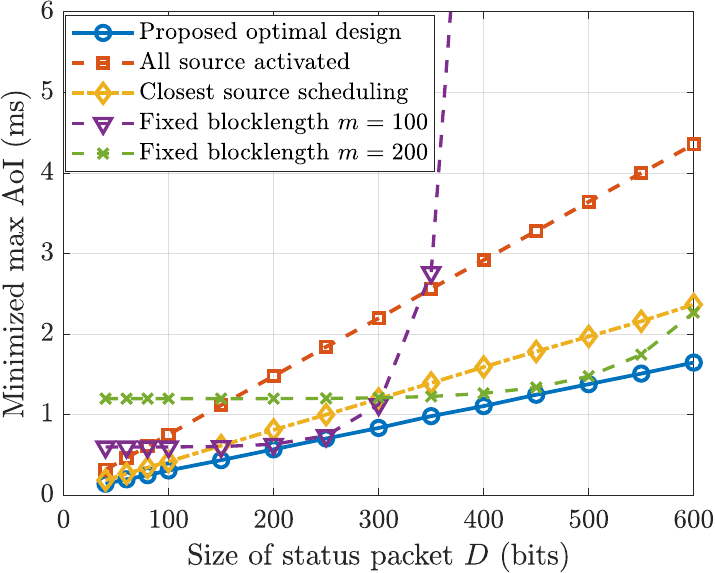}
\vspace{-.2cm}
    \caption{AoI performance comparison with benchmarks over different  $D$.}
    \label{fig: fig_wrt_D}
    \vspace{-.5cm}
\end{figure}

Finally, we evaluate the impact of the status packet payload~$D$ on the average AoI and benchmark the proposed scheme against the following baselines:
\begin{itemize}
    \item \emph{All source activated}: all available source node are activated, i.e., $\mathcal A_U=\mathcal A_S$, while source selection $\alpha(\cdot)$ and blocklength $m_{\alpha(k)}$ are still jointly optimized.
    \item \emph{Closest source scheduling}: each user~$k$ selects and activates only its geographically nearest source, while the blocklength $m_{\alpha(k)}$ is optimized. %ally designed.
    \item \emph{Fixed-blocklength}: the blocklength is fixed, % at some values 
    while $\mathcal A_S$ and $\alpha(\cdot)$ are optimized.
\end{itemize}
The achieved average AoI from different approaches is depicted in Fig.~\ref{fig: fig_wrt_D}. Our proposed design consistently achieves the lowest AoI across the entire range, which validates its optimality, and the relative gains over the baselines become more pronounced as $D$ grows. The \emph{all source activated} baseline incurs the strongest interference in SINR and hence the steepest AoI growth with $D$.
The \emph{closest source scheduling} is relatively close to the optimal selection, and could be a good heuristics. However, in cases where \emph{closest-source} would let each user activate a distinct source, the strong interference could deteriorate strongly the performance. The baseline \emph{fixed-blocklength} coincides with the optimum only in a small region (where the blocklength is close to the optimal blocklength), but performs poorly when other regions of~$D$.  
Overall, jointly optimizing $\mathcal A_U$, $\alpha(\cdot)$, and $m_{\alpha(k)}$ is essential when the payload size scales: ignoring any one of the three optimization variables yields much worse AoI.

\section{Conclusion}
In this paper, we considered a status update system with multiple source nodes observing the same status and multiple users equipped with fluid antennas. By applying adaptive fluid port selection at all users for simultaneous channel gain enhancement and interference mitigation, we formulated a maximum average AoI minimization problem and jointly optimized source scheduling, including source node activation and assignment, together with update scheduling at the activated source nodes. For a given source scheduling decision, we characterized the performance achievability of update scheduling based on fixed-point theory, which further enabled an efficient bisection algorithm for optimal update scheduling design. Subsequently, we established a criterion for achievable performance comparison over any set of source scheduling candidates, allowing efficient elimination of non-optimal source scheduling decisions. As a result, the globally optimal joint source and update scheduling solution was obtained. Simulation results validated the global optimality of the proposed solution and demonstrated the performance benefits of fluid antennas for AoI minimization. As future work, we will extend the proposed optimization framework to more complex multi-source scenarios involving multiple monitored statuses.

%\vspace*{-0.2cm}
%\bibliographystyle{IEEEtran}

\end{document}